\documentclass[9pt]{article}

\usepackage[margin=1in]{geometry}

\usepackage{setspace}
\usepackage{amsmath,amssymb,amsthm,mathtools}
\usepackage{yhmath}
\usepackage[normalem]{ulem}

\usepackage{graphicx}
\usepackage{booktabs}
\usepackage{caption}
\usepackage{subcaption}

\graphicspath{{fig/}}

\usepackage{tikz}
\usetikzlibrary{decorations.pathreplacing}

\usepackage{booktabs}

\usepackage{enumerate}

\usepackage[numbers]{natbib}   
\usepackage{hyperref}

\usepackage{xcolor}

\newcommand{\R}{\mathbb{R}}
\newcommand{\E}{\mathbb{E}}
\renewcommand{\P}{\mathbb{P}}

\newcommand{\I}{\mathcal{I}}
\newcommand{\X}{\mathcal{X}}
\newcommand{\Y}{\mathcal{Y}}
\renewcommand{\L}{\mathcal{L}}

\newcommand{\htu}{\hat\theta_u}
\newcommand{\htr}{\hat\theta_r}
\newcommand{\htl}{\hat\theta_\lambda}

\DeclareMathOperator*{\argmax}{arg\,max}
\DeclareMathOperator*{\argmin}{arg\,min}

\newtheorem{theorem}{Theorem}[section]

\newtheorem{prop}[theorem]{Proposition}

\newtheorem{definition}{Definition}

\counterwithin{figure}{section}
\counterwithin{equation}{section}
\counterwithin{table}{section}

\newcommand\ack[1]{
	\begingroup\renewcommand\thefootnote{*}\footnote{#1}
	\endgroup
}

\begin{document}

\title{
	From Utilitarian to Rawlsian Designs for Algorithmic Fairness\ack{
		I am grateful to Abraar Chaudhry, Miklos R\'acz, Roberto Rigobon, Matt Salganik, Till S\"{a}nger, and Philip Schnattinger for their comments and suggestions.
	} 
}

\author{
	Daniel E. Rigobon\thanks{
		\href{mailto:drigobon@princeton.edu}{\texttt{drigobon@princeton.edu}}
	}  \\
	\vspace{-0cm}
	{\small {ORFE, Princeton University} } \\
}

\date{\today}

\maketitle

\vspace{-0.8cm}
\begin{center}
	\begin{abstract}
		There is a lack of consensus within the literature as to how `fairness' of algorithmic systems can be measured, and different metrics can often be at odds. In this paper, we approach this task by drawing on the ethical frameworks of utilitarianism and John Rawls. Informally, these two theories of distributive justice measure the `good' as either a population's sum of utility, or worst-off outcomes, respectively. We present a parameterized class of objective functions that interpolates between these two (possibly) conflicting notions of the `good'. This class is shown to represent a relaxation of the Rawlsian `veil of ignorance', and its sequence of optimal solutions converges to both a utilitarian and Rawlsian optimum. Several other properties of this class are studied, including: 1) a relationship to regularized optimization, 2) feasibility of consistent estimation, and 3) algorithmic cost. In several real-world datasets, we compute optimal solutions and construct the tradeoff between utilitarian and Rawlsian notions of the `good'. Empirically, we demonstrate that increasing model complexity can manifest strict improvements to both measures of the `good'. This work suggests that the proper degree of `fairness' can be informed by a designer's preferences over the space of induced utilitarian and Rawlsian `good'.
	\end{abstract}
\end{center}
	
\maketitle

\section{Introduction}
Aspects of societal decision-making have become increasingly outsourced to algorithmic systems -- including criminal risk assessment \cite{Angwin2016Machine}, labor market organization \cite{Chalfin2016Productivity}, provision of medical care \cite{Kleinberg2015Prediction}, and more. The promise of these systems is often that they are more efficient than human decision-makers, and hence are appealing, such as for greater throughput or lower cost. However, under closer investigation, algorithmic systems have been seen to perpetuate or even amplify existing biases \cite{Angwin2016Machine, ONeil2017Weapons, Eubanks2018Automating}. These concerns have brought greater attention towards the design of algorithms that exhibit fairness or other ethical qualities \cite{Kearns2019Ethical, Barocas2021Fairness}.

A large body of existing work defines fairness through particular statistical or mathematical quantities. These measures can be used either as constraints \cite{Dwork2012Fairness, Hardt2016Equality}, or as a penalty imposed for deviating from equality \cite{Berk2017Convex}. Both methodologies take an egalitarian perspective, but one critique of this approach is that equality need not be fair \cite{Cooper2021Emergent}. Furthermore, even if this objection is suppressed, there are many different measures of fairness \cite{Narayanan2018Translation} that can conflict with each other \cite{Kleinberg2016Inherent}. An example of this was seen within the realm of criminal justice, where algorithm designers \cite{Dieterich2016COMPAS} demonstrated equality of one metric, while journalists \cite{Angwin2016Machine} criticized the algorithm's inability to satisfy another. Since the right measure of `fairness' can be unclear, some recent research instead seeks to provide moral and ethical justification behind particular measures \cite{Heidari2019Moral, Hertweck2021Moral}.

Beyond the difficulty of defining and developing a `fair' algorithm, there often exists a tension between the goals of model-builders and model-impacted individuals. For instance, algorithmic decision systems for targeted conditional cash transfer programs can exhibit comparatively greater accuracy than human-based systems, but both were found to yield inter-group inequalities \cite{Noriega-Campero2020Algorithmic}. On one hand, policymakers may view this result positively -- they increased coverage for the needy. On the other hand, those in need may themselves desire an allocation mechanism that is fair, equal, or other such qualities -- even at the expense of policymakers. From this tension has emerged a area of research aiming to understand a tradeoff between fairness and accuracy \cite{Diana2021Minimax, Little2022Fairness, Liang2022Algorithmic}. In particular, these studies adopt the perspective that fairness is \textit{inherently} at odds with accuracy, which is not trivially true.\footnote{See \citet{Cooper2021Emergent} for a more in-depth critique of common approaches to the `fairness-accuracy' tradeoff.}

Instead, this paper identifies and formalizes a tradeoff in algorithmic design between two different ethical frameworks of distributive justice: Utilitarian \cite{Bentham1996Introduction, Mill2008Utilitarianism} and Rawlsian \cite{Rawls2003Theory}. More precisely, we present a class of objective functions that interpolates between the preferences of a utilitarian and Rawlsian designer.

In many algorithmic settings, a model is tasked with distributing some quantity of predictive loss throughout a population. Each of these two approaches to distributive justice can be used to determine which model's allocation of loss is most `good'. The utilitarian paradigm is often associated with accuracy (as opposed to fairness), but this need not be the case. If, for example, a utilitarian believed that each individuals' disutility is proportional to their model-induced squared error, then they would argue that the most accurate (with respect to mean squared error) model is also maximally `good'. In doing so, moreover, this utilitarian designer weighed the needs of all individuals equally -- could this not be `fair'? It is therefore critical to emphasize that such statements about what `fairness' is (or is not) must therefore reflect a contextual acceptance (or rejection) of particular ethical frameworks. Namely, the assertion of a `fairness-accuracy' tradeoff requires that either: 1) utility is not tied to accuracy or 2) a utilitarian approach to the problem is unfair.

We make no such assertion in this work. A utilitarian designer will define `good' to be the sum of population utility, whereas a Rawlsian designer will measure `good' through the outcomes of a population's least advantaged.\footnote{In Section~\ref{sec:model} we will discuss these two theories and their implied objective functions in more detail.} The objective functions in this paper therefore arguably reflect a `fairness-fairness' tradeoff -- or to be precise, a `Utilitarian good-Rawlsian good' tradeoff. In part, this tradeoff is valuable to understand because each ethical framework addresses a common critique of the other. A utilitarian can be indifferent towards inequality, whereas a Rawlsian's greatest concern is the most needy. Conversely, while a Rawlsian designer is unconcerned with the preferences of the majority, a utilitarian weights all individuals' preferences equally. This paper's approach allows us to partially address the shortcomings of each framework while leveraging their advantages.

Our main contributions are threefold. First, we conceptualize a class of objective functions and show that they capture a relaxation of Rawls's `original position'. This result exhibits close ties to social welfare and risk aversion. Second, we study convergence properties of the objective functions and their minimzers. These technical results verify that we are indeed interpolating between: 1) utilitarian and Rawlsian measures of `good', and 2) their most desirable outcomes. Finally, our experiments show the tradeoff between these two measures on several common datasets, and demonstrate how this tradeoff is influenced by model complexity. In particular, a designer's preferences (over bundles of Rawlsian and utilitarian `good') can be used to determine their desired point along this tradeoff. In these experiments, we also study group-averaged loss, and see that an egalitarian approach may be significantly at odds with Rawlsian principles.

The rest of the paper is organized as follows. Section~\ref{sec:lit} reviews the most relevant and recent work. Section~\ref{sec:model} presents the learning problem and objective functions for utilitarian and Rawlsian designers. Section~\ref{sec:continuum} contains our main conceptual and theoretical results, where we introduce a class of objective functions and study its properties. Section~\ref{sec:experiments} trains various models on real-world datasets and studies several aspects of their performance. Finally, Section~\ref{sec:discussion} concludes and presents directions for future work.

\subsection{Relevant Literature}
\label{sec:lit}
There are several areas of work related to this paper, each of which we present here. However, we introduce and discuss the relevant ethical theories in Section~\ref{sec:model}.

A significant branch of literature seeks to measure fairness through mathematical or statistical measures. These works address fairness by imposing constraints or penalties based on these measures during the in-processing stage of model building, e.g. in \cite{Dwork2012Fairness, Hardt2016Equality,Berk2017Convex, Corbett-Davies2017Algorithmic}. However, there is not a universally agreed upon measure of fairness. Moreover, such formal criteria for fairness can conflict \cite{Kleinberg2016Inherent}, yield to long-term damage \cite{Liu2018Delayed} or are subject to fundamental statistical limitations \cite{Corbett-Davies2018Measure}. As a response to these challenges, recent work has grounded particular measures of fairness in moral and ethical arguments \cite{Heidari2019Moral, Hertweck2021Moral}. The greatest similarity between this area of work and our paper is a shared approach to fairness through distinct theories of ethical `good'.

Rawls's framework has appeared in computer science literature through minimax fairness \cite{Heidari2018Fairness, Martinez2020Minimax, Lahoti2020Fairness, Diana2021Minimax, Papadaki2022Minimax, Yang2022Minimax, Little2022Fairness}. These papers focus largely on \textit{group} minimax fairness. Instead, we study \textit{individual} minimax fairness through a relaxation of the Rawlsian `original position'. Two comparative advantages of our approach are: 1) avoiding any danger of fairness gerrymandering (see \cite{Kearns2018Preventing} for another solution to this issue), and 2) no requirement to be given group labels (see \cite{Hashimoto2018Fairness, Lahoti2020Fairness} for other such approaches). The most similar paper to our own is \cite{Heidari2018Fairness}, where the authors use a closely related social welfare function to constrain an accuracy-maximizing optimization problem. However, we do not consider accuracy to be the fundamental objective -- instead focusing on maximizing social welfare itself. Finally, we note that a different principle from Rawls's theory of justice has been studied by \citet{Liu2021RAWLSNET}, who provide techniques for imposing fair equality of opportunity on Bayesian graphical models.

Social choice theory and welfare economics have also been influenced by Rawls's principles \cite{Sen1976Welfare, Hammond1976Equity, DAspremont1977Equity}. However, differing views in these areas can argue that only utilitarian designs are possible \cite{Maskin1978Theorem} or rational \cite{Harsanyi1975Nonlinear}. We note that the class of social welfare functions that appear in this paper and \cite{Heidari2018Fairness} are justified axiomatically by \citet{Roberts1980Interpersonal}. As a result, recent approaches to fairness in machine learning relying on notions of social welfare \cite{Rambachan2020Economic, Rambachan2021Economic} are closely related to our work. Namely, this paper's approach can be interpreted as a planner aiming to maximize social welfare for a particular class of functions. However, in contrast to \citet{Rambachan2021Economic}, the social welfare functions in this paper do not require group-specific weights to be given a priori, and instead rely on a designer's degree of risk aversion from the Rawlsian `original position'.

Minimax optimizations (or variants thereof) have been well-studied for their robustness qualities. For example, distributionally-robust optimization (DRO) problems can significantly improve predictive outcomes for underrepresented groups \cite{Hashimoto2018Fairness,Sagawa2020Distributionally, Li2021Evaluating}. In addition, \citet{Lahoti2020Fairness} use a variant of DRO that adversarially weights observations during the learning process to improve the performance of worst-off groups, which is closely tied to the Rawlsian notion of `good'. Most related to our paper is the identical relaxation of minimax optimization known as `Tilted Empirical Risk Minimization' (TERM), proposed by \citet{Li2021Tilted}. The authors study properties of both the loss function and its optimal solutions under the assumption of generalized linear models. One core difference is conceptual -- we focus on representing these objective functions as relaxations of the Rawlsian veil of ignorance, incorporating features of risk aversion. Technically, we rigorously prove a stronger convergence property for the optimal solutions. Finally, in our experiments we study the impact of increasing model complexity on the tradeoff between two notions of the `good'.

Finally, a number of papers in the literature aim to characterize a `fairness-accuracy' tradeoff \cite{Liang2022Algorithmic, Little2022Fairness}. We note that this tradeoff also appears in several fairness-constrained approaches \cite{Corbett-Davies2017Algorithmic, Diana2021Minimax}. \citet{Cooper2021Emergent} present a critique of these studies, questioning the assumptions that fairness and accuracy are at odds, that equality is fair, and more. We empirically study the existence of a similar tradeoff, but arguing that it instead reflects a balance between utilitarian and Rawlsian measures of `good'. In addition, we observe how the tradeoff is affected by changes to model complexity, which to the best of our knowledge, has not been previously studied.

\section{Modeling and Ethical Frameworks}
\label{sec:model}
In this section, we describe a general supervised learning setting, and two theories of distributive justice that can be used to perform model selection.

Let $\{x_i, y_i\}_{i=1...n}$ denote a set of observations, where $x_i \in \X$ are features and $y_i \in \Y$ is a target. In simple classification settings, $\Y = \{0,1\}$. A set of candidate models is defined by $f_\theta: \X \to \Y$ for each $\theta$ in parameter space $\Theta$. Finally, let $\ell \left( f_\theta(x_i), y_i \right)$ denote the loss incurred by model $f_\theta$ on observation $i$. For simplicity of notation, we occasionally write $\ell_i(\theta) = \ell(f_\theta(x_i), y_i)$. The loss function $\ell$ is assumed to be primitive (i.e. given), and we will assume that $\ell(f_\theta(x_i), y_i)$ represents the disutility experienced by an individual with characteristics $x_i, y_i$ under model $f_\theta$.

\paragraph{Utilitarian}
A utilitarian designer, aligning with the political philosophy of John Stuart Mill \cite{Mill2008Utilitarianism} and Jeremy Bentham \cite{Bentham1996Introduction}, would seek to minimize the sum of population disutility -- equivalently maximizing total utility.\footnote{A utilitarian need not define `utility' as equal or proportional to negative loss. However, in principle $\ell$ can represent any desired metric of damage. We need only that this same metric applies to universally to all individuals. The construction of such metrics is far beyond the scope of this paper, and largely driven by each model's application domain.} A fundamental feature of utilitarian ethics is that it can justify hurting one or more individuals if, in doing so, others in the population are sufficiently compensated. Utilitarianism also exhibits the valuable property that all individuals' needs are equally important. However, it is blind to higher-order characteristics of the distribution of utilities -- i.e. large increases to its variance are justified in the name of an infinitesimal increase to its mean. As a result, it has been criticized for its indifference towards inequality \cite{Sen1979Equality}. In this paper, a \textit{utilitarian} designer solves:
\begin{equation}
	\label{opt:utilitarian}
	\min_{\theta\in\Theta} \sum_i \ell( f_\theta(x_i), y_i ),
\end{equation}
whose minimizer is given by $\htu$. In many common machine learning examples, the objective function may equal $n^{-1}\sum_i (y_i - f_\theta(x_i))^2$, which corresponds to a utilitarian designer with disutility equal to squared loss $(y_i - f_\theta(x_i))^2$.

\paragraph{Rawlsian}
Another possible approach comes from the thought experiment and philosophy of John Rawls \cite{Rawls2003Theory}. For completeness, we briefly state a few of his main points. First, Rawls presents the `original position' (also referred to as the `veil of ignorance') -- wherein individuals do not know their place in society, talents, or even notions of what entails a good life. From such a position, he argues that a rational individual would desire that ``inequalities are to be arranged so that they are [...] to the greatest benefit of the least advantaged'' \cite{Rawls2003Theory}.\footnote{A comprehensive summary of Rawls's philosophy is far beyond the scope of this paper. However, we note that Rawls prioritizes two other principles before the one mentioned here. First, that all individuals are entitled to the greatest possible set of individual liberties. This principle retains a minimax flavor -- if the individual with least liberties agrees to a particular organization of society, then behind the veil of ignorance, all others would necessarily agree. A second principle is that offices yielding any inequalities must be equally accessible to all (i.e. equality of opportunity). The latter principle has featured in several recent studies \cite{Hardt2016Equality, Liu2021RAWLSNET}. \label{foot:rawls}} Termed the `difference principle', this minimax approach is desirable through its ability to address utilitarianism's indifference to inequality. However, it is an extremely strict paradigm, and is largely unconcerned with the majority of a population. In the context of this paper, a \textit{Rawlsian} designer would aim to solve:
\begin{equation}
	\label{opt:rawlsian}
	\min_{\theta\in\Theta} \max_i \ell(f_\theta(x_i), y_i),
\end{equation}
whose minimzer is denoted $\htr$. An immediate observation of problem~\eqref{opt:rawlsian} is that the objective function is non-differentiable, and therefore can be difficult to solve in practice. In addition, a Rawlsian designer's optimal model necessarily satisfies one of two criterion: 1) if the maximization in \eqref{opt:rawlsian} has a unique maximizer $i$, then the model has reached the fundamental limit of predictability for observation $i$, 2) if there are multiple maximizers, it is impossible to reduce the loss for one of these without increasing the loss of another in doing so. Effectively, this means the optimal model is agnostic to any easily-predictable observations. A Rawlsian designer therefore views outliers in a fundamentally different manner from traditional data scientists -- outliers represent their assessment of good, and are not noise to be discarded. We also note that problem \eqref{opt:rawlsian} is also closely tied to robust control design in engineering \citep[Chapter 14]{Kemin1998Essentials} and optimization under ambiguity in economics \cite{Gilboa1989Maxmin} -- the latter of which is reminiscent of the original position.

We remark that our presentation of Rawls's philosophy is greatly simplified. In order to justify that the solution to \eqref{opt:rawlsian} is `fair', Rawls would first require that principles be satisfied: 1) basic liberties are guaranteed and 2) offices carrying inequalities are open and equally accessible to all (see footnote~\ref{foot:rawls}). The latter principle on fair equality of opportunity has appeared in several recent papers \cite{Hardt2016Equality, Heidari2019Moral, Liu2021RAWLSNET}. While in this paper we assume that both principles hold, systemic inequalities in society would suggest that this need not be the case. A more complete integration of Rawls' principles into the design and implementation of algorithmic systems remains a rich area for future work. It is also important to note that \eqref{opt:rawlsian} reflects the result of applying Rawls's difference principle to a \textit{relaxed} original position. Individuals must at least know their notion of good (i.e. the function $\ell$), but still be unaware of all other characteristics (i.e. covariates $x$ and target $y$). It is a much deeper question to study if such a relaxation still yields the same principles of justice.

Much of the literature on algorithmic fairness is interested in \textit{group} measures of fairness. Let the given observations be partitioned into groups $G_1,...,G_m$, which may not be mutually exclusive. Applying a minimax approach to average group loss would give:
\begin{equation}
	\label{opt:grp_rawlsian}
	\min_{\theta\in\Theta} \max_{j} \frac{1}{|G_j|} \sum_{i \in G_j} \ell(f_\theta(x_i), y_i).
\end{equation}
Notice that this is an inter-group Rawlsian paradigm coupled with intra-group utilitarianism.

It is not immediately clear which of \eqref{opt:rawlsian} or \eqref{opt:grp_rawlsian} is preferred. Indeed, there is a contentious debate in the literature between \textit{individual} and \textit{group} fairness.\footnote{See \citet{Dwork2012Fairness} and \citet{Sharifi-Malvajerdi2019Average} for examples of individual fairness, or \citet{Hardt2016Equality} and \citet{Diana2021Minimax} for group fairness. Also see \citet{Kearns2018Preventing} on mixing both individual and group notions of fairness.} Individual fairness represents a limiting case of group fairness, but it can generalize poorly and be difficult to measure. Conversely, group fairness can fail to account for intra-group differences in outcomes, leading to so-called `fairness gerrymandering' \cite{Kearns2018Preventing}.

We do not aim to resolve this debate, only to argue that the individualized approach in \eqref{opt:rawlsian} is closer to reflecting Rawls's original position than \eqref{opt:grp_rawlsian}. In practice it is impossible to perfectly manifest Rawls's original position -- recall that the individual-driven fairness of \eqref{opt:rawlsian} is still a relaxation of the true veil of ignorance. An individual merely being within the sample may reflect certainty about some of their characteristics, e.g. that they are applying for a low-paying job, high-interest loan, or have been previously incarcerated. However, they remain uncertain of their characteristics within the sample -- including group membership and the distribution of characteristics within each group. Now, we can instead imagine a different relaxation of the original position that is related to the group-wise approach of \eqref{opt:grp_rawlsian}. Here, the group-conditional distributions of covariates $X,Y$ must be known, while only group membership is uncertain. Individuals in this new position face strictly less uncertainty than before, and hence the veil of ignorance is more transparent. Although this paper focuses on individual fairness, we also empirically study the effects on groups.

\section{Utilitarian-Rawlsian Continuum}
\label{sec:continuum}
Ultimately, the approach of both utilitarian and Rawlsian designers can have shortcomings. In the main conceptual contribution of this paper, we define a set of objective functions that interpolates between these two seemingly conflicting paradigms.

\begin{definition}[Utilitarian-Rawlsian Objective, $L(\theta; \lambda)$]
	For any $\lambda \in (0,\infty)$, we define:
	\begin{equation}
		\label{eq:continuum}
		L(\theta;\lambda) = \frac{1}{\lambda} \log\left( \frac{1}{n} \sum_i e^{\lambda \ell(f_\theta(x_i), y_i)} \right).
	\end{equation}
\end{definition}

The optimization problem associated with this objective is therefore:
\begin{equation}
	\label{opt:continuum}
	 \min_{\theta \in \Theta} L(\theta;\lambda) = \min_{\theta\in\Theta} \frac{1}{\lambda} \log\left( \frac{1}{n} \sum_i e^{\lambda \ell(f_\theta(x_i), y_i)} \right),
\end{equation}
whose optimal solution is $\htl$. Main theoretical results in Section~\ref{sec:main_prop} show that problem \eqref{opt:continuum} is a continuous relaxation between \eqref{opt:utilitarian} and \eqref{opt:rawlsian} -- which is studied through both the Rawlsian original position, and convergence properties of both \eqref{eq:continuum} and its optimal solutions. Interesting connections to a regularized fairness approach and further properties of \eqref{opt:continuum} are briefly presented in Section~\ref{sec:further_prop}.

For completeness, we include the following analogous relaxation of the group-wise minimax approach in \eqref{opt:grp_rawlsian}.
\begin{equation}
	\label{opt:grp_continuum}
	\min_{\theta\in\Theta} \frac{1}{\lambda} \log\left( \frac{1}{m} \sum_j e^{\lambda |G_j|^{-1} \sum_{i \in G_j} \ell(f_\theta(x_i), y_i)} \right).
\end{equation}
However, as previously justified, the remainder of this paper focuses exclusively on $L(\theta; \lambda)$ and problem \eqref{opt:continuum}.

\subsection{Characterization}
\label{sec:main_prop}
This section contains our main theoretical results. We present an interpretation of $L(\theta; \lambda)$ and problem \eqref{opt:continuum} that reflects a weakened notion of the Rawlsian veil of ignorance, and note important connections to social welfare maximization and risk aversion. Next, we study convergence properties of both $L(\theta; \lambda)$ and its minimizers $\htl$ . Both results together verify that we are indeed representing a continuum of objective functions between utilitarian and Rawlsian designs. Finally, we conclude by briefly analyzing a simple setting -- univariate linear regression.

\paragraph{Relaxed Veil of Ignorance and Welfare}
First, we show that problem \eqref{opt:continuum} captures a natural relaxation of Rawls's  original position. Consider an individual who is randomly assigned covariates $X$ and `true' target $Y$ according to $X, Y \sim \mathrm{Unif}(\{(x_1,y_1),...,(x_n,y_n)\})$. In each state of the world $\theta$, she observes some random loss $\ell(f_\theta(X), Y)$. If this loss carries disutility proportional to $e^{\lambda \ell(f_\theta(X), Y)}$, then it is possible to see that:
\begin{equation}
	\frac{1}{n}\sum_i e^{\lambda \ell(f_\theta(x_i), y_i)} = \E_{X,Y}\left[ e^{\lambda \ell(f_\theta(X),Y)} \right].
\end{equation}
In this context, the solution to problem \eqref{opt:continuum} is equivalently minimizing expected disutility of loss for an individual with constant absolute risk aversion $\lambda$. Informally, $\lambda$ captures the degree to which she dislikes uncertainty in the distribution of $\ell(f_\theta(X),Y)$. This connection is seen in the following Proposition, which is presented without proof.

\begin{prop}
	\label{prop:exp_util}
	Let $u_\lambda(\ell_i(\theta)) = - e^{ \lambda \ell(f_\theta(x_i), y_i)}$ denote the utility function of individual $i$ corresponding to model $f_\theta$. Then:
	\begin{equation}
		\label{eq:exp_util}
		\argmin_{\theta \in \Theta} L(\theta; \lambda) = \argmax_{\theta \in \Theta} \E_{i \sim \mathrm{Unif}[1...n]} \left[ u_\lambda(\ell_i(\theta)) \right].
	\end{equation}
\end{prop}

Notice that the connection established in Proposition~\ref{prop:exp_util} implies that the function $L(\theta; \lambda)$ is \textit{effectively} utilitarian -- up to a monotone transformation, it is proportional to the total utility in the population. However, it is utilitarian with respect to a particular measure of utility -- not loss itself. Since $u_\lambda$ is a non-linear function of loss, the optimal solution to \eqref{eq:exp_util} does not coincide with the utilitarian optimum.

Another observation made regarding the right-hand side of \eqref{eq:exp_util}, drawn from social welfare and choice theory, is its independence of common \textit{level}. Namely, for any constant $\beta$, it is easy to see that
\begin{equation}
	\label{eq:ind_level}
	\sum_i u_\lambda(\ell_i(\theta_1)) \ge \sum_i u_\lambda(\ell_i(\theta_2)) 
	\iff
	\sum_i u_\lambda(\ell_i(\theta_1) + \beta) \ge \sum_i u_\lambda(\ell_i(\theta_2) + \beta) .
\end{equation}
As a consequence of this property, only absolute differences between individual losses affect the optimal solution. However, \eqref{eq:ind_level} would not hold if instead of adding, we multiplied all losses by some $\beta > 0$. This means that the expectation in \eqref{eq:exp_util} does not exhibit independence of common \textit{scale}. To satisfy this latter property, it would be necessary to use a different form of $u_\lambda(\cdot)$ -- see \cite{Heidari2018Fairness} for functions satisfying only independence of common scale, or \cite{Harsanyi1955Cardinal, Maskin1978Theorem} for those satisfying both independence of common scale and level. Furthermore, from the axiomatic characterization of \cite{Roberts1980Interpersonal}, we can see that $L(\theta; \lambda)$ is a monotone transformation of a particular class of social welfare functions. As a result, problem \eqref{opt:continuum} represents a social welfare-maximizing approach towards algorithmic fairness.


\paragraph{Convergence}
We now turn to the main technical results of this paper. For limiting values of $\lambda$, we study the behavior of $L(\theta; \lambda)$ and the optimal solutions to problem \eqref{opt:continuum}.

As $\lambda \to \infty$, the sum in \eqref{eq:continuum} is dominated by the observation with maximum loss, and hence approaches the Rawlsian minimax objective in \eqref{opt:rawlsian}. Conversely, as $\lambda \to 0$ the exponential is approximately linear in its argument, which leads directly to the utilitarian objective of \eqref{opt:utilitarian}. The following result shows that for any $\theta$, $L(\theta; \lambda)$ indeed satisfies these properties.

\begin{prop}
	\label{prop:pointwise_convergence}
	For all $\theta \in \Theta$: 
	\begin{alignat*}{3}
		\lim_{\lambda \to 0} & L(\theta; \lambda)&&=\frac{1}{n} \sum_i \ell(f_\theta(x_i), y_i) \\
		\lim_{\lambda \to \infty} & L(\theta; \lambda)&&=\max_i \ell(f_\theta(x_i), y_i).
	\end{alignat*}
\end{prop}

The proof is found in Appendix~\ref{app:pfs}. Although simple, this result on pointwise convergence verifies that at small (resp. large) values of $\lambda$, the objective function in problem \eqref{opt:continuum} behaves exactly like that of \eqref{opt:utilitarian} (resp. \eqref{opt:rawlsian}). Therefore, it is interpolating between utilitarian and Rawlsian measures of good.

In fact, it is possible to show that $L(\theta; \lambda)$ exhibits a stronger form of convergence, which can yield convergence of its minimizers. This is formalized in the following main result.

\begin{theorem}
	\label{thm:argmin_convergence}
	Let $\htl$ be the optimal solution to \eqref{opt:continuum}. If $\Y$ is compact, the set $\{ f_\theta(x), \, \theta \in \Theta\}$ is compact for all $x \in \X $, and $\ell(\cdot, \cdot)$ is continuous, then:
	\begin{equation}
		\begin{alignedat}{3}
			\lim_{\lambda \to 0} & \htl &&\in \argmin_{\theta \in \Theta} \frac{1}{n}\sum_i \ell(f_\theta(x_i), y_i) \\
			\lim_{\lambda \to \infty} & \htl &&\in \argmin_{\theta \in \Theta} \max_i \ell(f_\theta(x_i), y_i).
		\end{alignedat}
	\end{equation}
	In addition, if the minimizers on the right-hand side are unique (denoted $\htu$ and $\htr$), then:
	\begin{equation}
		\begin{alignedat}{3}
			\lim_{\lambda\to0} &\htl && = \htu \\
			\lim_{\lambda\to\infty} &\htl && = \htr.
		\end{alignedat}
	\end{equation}
\end{theorem}

The proof in Appendix~\ref{app:pfs} uses the notion of $\Gamma$-convergence for a sequence of functions -- which is stronger than uniform convergence. It can be leveraged to characterize the sequence of their minimizers \cite{Braides2006Handbook, Maso2012Introduction}.

Theorem~\ref{thm:argmin_convergence} is useful for several reasons. First, it further justifies the use of $L(\theta; \lambda)$ for capturing both utilitarian and Rawlsian optimal designs. In addition, it shows that \textit{some} minimax solution can be approximated by a sequence of minimizers to the relaxed problems. In the case where $\htr$ is not unique, then we conjecture that it is possible to characterize the limit of $\htl$ more precisely as follows.

For $u\in \R^n$ let $u^{(1)}$ denote its largest entry, and $u^{(-1)}$ be the vector of remaining entries. For $u,v \in \R^n$, we say that $u \preceq v$ if $u^{(1)} < v^{(1)}$ or both $u^{(1)} = v^{(1)}$ and $u^{(-1)} \preceq u^{(-1)}$. This is often known as the \textit{leximax} ordering. We expect that $\lim_{\lambda \to \infty} \htl = \hat\theta_{lex}$, where $\ell(\hat\theta_{lex}) \preceq \ell(\theta)$ for all $\theta$, but to the best of our knowledge this has not yet been rigorously proven.


\paragraph{Example: Linear Regression}
We now turn to a simple setting, with $\Theta = \R$, $f_\theta(x) = \theta x$, and $\ell(\hat y, y) = (\hat y - y) ^2$. For simplicity we also assume that $\E[X] = \E[Y] = 0$. Plugging these into problem \eqref{opt:continuum} yields the following convex and unconstrained optimization problem:
\begin{equation}
	\label{opt:lin_reg}
	\min_{\theta \in \R} \frac{1}{\lambda} \log \left(\frac{1}{n} \sum_i e^{\lambda (\theta x_i - y_i)^2}\right).
\end{equation}
The necessary (and sufficient) first-order condition can be computed as:
\begin{equation}
	0 = \sum_i \frac{e^{\lambda (\htl x_i - y_i)^2}}{\sum_j e^{\lambda (\htl x_j - y_j)^2}} (\htl x_i - y_i) x_i.
\end{equation}
Manipulating the above, we obtain:
\begin{equation}
	\htl = \frac{\widetilde{\mathrm{Cov}}(X,Y)}{\widetilde{\mathrm{Var}}(X)},
\end{equation}
which seems to be the usual estimator. However, now the covariance and variance are computed with respect to a twisted measure $\tilde \P_\lambda$, which satisfies $\frac{d\tilde \P_\lambda}{d\P}(x_i) = \frac{e^{\lambda (\htl x_i - y_i)^2} }{n^{-1} \sum_j e^{\lambda (\htl x_j - y_j)^2}}.$ Namely, this measure ascribes larger (resp. smaller) weights to observations whose exponentiated loss is greater (resp. less) than the average. However, it depends explicitly on $\htl$, and therefore the optimal solution cannot be computed in closed form.

Let us now informally consider what happens for large $\lambda$. The quantity $\sum_j e^{\lambda (\htl x_j - y_j)^2}$ is dominated by the observations with maximum loss, and equal measure is given to each of them. Therefore, if we let $\I = \argmax_i (\htl x_i - y_i)^2$, then $\widetilde{\mathrm{Cov}}(X,Y) \approx \sum_{i \in \I} y_{i} x_{i}$ and $\widetilde{\mathrm{Var}}(X) \approx \sum_{i \in \I}x_{i}^2$. Hence, for large $\lambda$, it follows that $\htl \approx \frac{\mathrm{Cov}(X_{\I}, Y_{\I})}{\mathrm{Var}(X_{\I})}$, which is exactly the usual least squares estimator -- only restricted to observations in the set $\I$.

We note that this setting has been more closely studied in another paper. In particular, under the assumption of generalized linear models, \citet{Li2021Tilted} derive several interesting properties of the optimal solution. Under reasonable conditions, they show that the average loss (resp. maximum loss) is increasing (resp. decreasing) in $\lambda$ at the optimal solution $\htl$. In addition, they prove that the empirical variance of the residuals $(\htl x_i - y_i)$ is non-increasing in $\lambda$, and verifies this to be the case in simulations. We might therefore expect that the finite-sample variance of $\htl$ is also non-decreasing in $\lambda$, although this has not been formally shown.


\subsection{Further Properties}
\label{sec:further_prop}

There are many other desirable properties of optimal solutions to learning problems, including (but not limited to) generalization performance, estimator properties, computational tractability, and optimality guarantees. In this section, we briefly touch on some of these topics and highlight connections to other areas of work -- such as fairness-penalized optimization and adversarial reweighting of observations. Strengthening these results remains an active and interesting directions for future research.


\paragraph{Identifiability}
From a statistical perspective, a natural question to ask about problem \eqref{opt:continuum} is whether or not the `true' parameter is identifiable. That said, if the data is generated according to some $\theta^* \in \Theta$, is it possible to find $\theta^*$? In the following result, we show that this requires a stronger condition than unbiased errors, which depends on the choice of loss function.

\begin{prop}
	\label{prop:identifible}
	Assume that $\exists \theta^* \in \Theta$ for which $Y_i | X_i \overset{i.i.d.}{\sim} f_{\theta^*}(X_i) + \epsilon_i$, where $\epsilon_1...\epsilon_n | X_i$ are i.i.d. according to density function $f_\epsilon$. Assume also that $\ell(\hat y, y) = g(y - \hat y)$ for some differentiable and positive-valued $g$, that is strictly increasing in $|y - \hat y|$. For any $r \in \mathrm{Range}(g)$, let $g^{-1}_{(-)}(r)$ and $g^{-1}_{(+)}(r)$ denote its negatively- and positively-valued inverse, respectively. 
	
	Then, if and only if $g'\left(g^{-1}_{(+)}(r)\right) f_\epsilon\left(g^{-1}_{(+)}(r)\right) = -g'\left(g^{-1}_{(-)}(r)\right) f_\epsilon\left(g^{-1}_{(-)}(r)\right)$ for all $r \in \mathrm{Range}(g)$, then over the randomness of the sample $X,Y$, we have: 
	\begin{equation}
		\E\left[\nabla_\theta L(\theta^*; \lambda, X, Y)\right] = 0, \, \forall \lambda > 0.
	\end{equation}
\end{prop}

The proof is found in Appendix~\ref{app:pfs}. A special case of Proposition~\ref{prop:identifible} occurs when both the distribution of errors $f_\epsilon$ and the primitive loss function $\ell$ are symmetric. In particular, given a symmetric loss function $\ell$, if the distribution of errors is not symmetric, then there is no hope of obtaining a consistent estimator -- the true parameter $\theta^*$ is not identifiable through the first-order conditions. However, in order to definitely prove consistency, it may be necessary to show that $L(\theta; \lambda)$ satisfies a uniform law of large numbers, which often requires compactness of $\Theta$ and that $L(\theta; \lambda)$ be bounded by a integrable function.


\paragraph{Regularization}
We now show that problem \eqref{opt:continuum} can be used to bound an optimization problem that penalizes the objective based on its worst-case individual loss. Since for any $\lambda$ and  $\theta$, $L(\theta; \lambda)$ is upper bounded (resp. lower bounded) by the maximum (resp. average) loss, there must exist some $\gamma \in (0,1)$ for which:
\begin{equation}
	\label{eq:l_g_equality}
	L(\theta; \lambda) = \gamma \left(\frac{1}{n} \sum_i \ell(f_\theta(x_i), y_i)\right) + (1-\gamma) \max_i \ell(f_\theta(x_i), y_i).
\end{equation}
Fix some $\lambda$ and let $\htl$ be the associated optimal solution to problem \eqref{opt:continuum}. We can compute its corresponding value of $\hat\gamma$, and majorize the following penalized optimization problem:
\begin{equation}
	\label{opt:regularized}
	\min_{\theta\in\Theta} \, \frac{1}{n} \sum_i \ell(f_\theta(x_i), y_i) + \frac{1-\hat\gamma}{\hat\gamma} \max_i \ell(f_\theta(x_i), y_i) \le \frac{1}{\hat\gamma} L(\htl; \lambda).
\end{equation}
A similar bound can be computed in the opposite order: fix $\gamma$, minimize the $\gamma$-regularized objective (that appears in the left-hand side of \eqref{opt:regularized}) for $\hat\theta_\gamma$, compute the value of $\hat\lambda$ that satisfies \eqref{eq:l_g_equality}, and observe that:
$$\frac{1}{\gamma} \min_{\theta \in \Theta} L(\theta; \hat\lambda) \le \frac{1}{n} \sum_i \ell\left(f_{\hat\theta_\gamma}(x_i), y_i\right) + \frac{1-\gamma}{\gamma} \max_i \ell\left(f_{\hat\theta_\gamma}(x_i), y_i\right).$$
We note that equality in the above need not hold -- optimal solutions to problem \eqref{opt:continuum} need not be minima of \eqref{opt:regularized}. In particular, $L(\theta; \lambda)$ depends on the full distribution of $\ell(\theta)$, whereas problem \eqref{opt:regularized} is only concerned with its mean and lowest percentile. Hence, it is not always the case that problem \eqref{opt:continuum} yields the value of $\theta$ that minimizes worst-case loss for some fixed average loss.\footnote{Consider a simple example where there are two possible loss profiles (i.e. two possible values for $\ell_1(\theta),...,\ell_n(\theta)$) given by $[0.5,2.75,2.75]$ and $[1,2,3]$. The former has smaller (resp. larger) objective value for small (resp. large) $\lambda$. However, both have the same average loss. In particular, it is possible for $[0.5,2.75,2.75]$ to not be the minimizer of $L(\cdot; \lambda)$.} Nonetheless, a comparative benefit of problem \eqref{opt:continuum} is that the objective function is smooth, and therefore can be solved numerically by many common algorithms.


\paragraph{Algorithmic Considerations}
From a technical perspective, $L(\theta; \lambda)$ may be preferable to the Rawlsian minimax objective because it is both differentiable and convex, which is shown in the following.

\begin{prop}
	\label{prop:diff_convex}
	If $\ell(f_\theta(x), y)$ is differentiable and convex in $\theta$ for all $x,y$, then $L(\theta; \lambda)$ is convex.
\end{prop}

\begin{proof}
	Since $z \to \log ( \sum_i e^{\lambda z_i} )$ is convex (in $z\in \R_{\ge 0}^n$) and non-negative, then its composition with $\ell(f_\theta(x_i), y_i)$ (which needs only be differentiable and convex) is also convex.
\end{proof}

As a result, we can use first-order optimization methods, which often have guaranteed convergence to a local minimum. Computing the gradient of $L(\theta; \lambda)$ gives:
\begin{equation}
	\label{eq:grad_relaxed}
	\nabla_\theta L(\theta; \lambda) = \sum_i  \frac{e^{\lambda \ell(f_\theta(x_i), y_i)} }{ \sum_k e^{\lambda \ell(f_\theta(x_k), y_k)} } \nabla_{\theta}\ell(f_\theta(x_i), y_i),
\end{equation}
where $\nabla_\theta \ell(f_\theta(x_i), y_i)$ denotes the full gradient of $\ell(f_\theta(x_i), y_i)$ with respect to $\theta$.\footnote{For brevity, we omit the gradient of $f_\theta$ that would appear from the chain rule.} Observe that this is simply a weighted average of the gradient at each observation $i$, where the weights are positively correlated with the losses. There is a relationship to adversarially re-weighted learning, for example, \citet{Lahoti2020Fairness} allows an adversarial agent to re-weight observations in order to increase a learner's weighted loss. Here, the weights are similarly related to loss, only not adversarial but pre-determined.

In effect, steps along the gradient in \eqref{eq:grad_relaxed} reflect a relaxed version of Rawls's difference principle. Originally, the principle permits inequalities only when they are to the benefit of the least advantaged. Therefore, to `improve' over the status quo, one should aim to assist the worst-off. In \eqref{eq:grad_relaxed}, this is not necessarily the case -- any harm done to the worst-off can be justified if there is \textit{sufficient} benefit provided to others. The ability for such a setting to arise reflects a fundamental utilitarian influence. However, as $\lambda$ grows, it becomes increasingly (and impossibly) difficult to justify any harm done to the worst-off.

Practically, there can be a significant computational cost associated with gradient descent. The following Proposition is from Theorem 13 in \cite{Li2021Tilted}, and slightly re-formulated here.

\begin{prop}
	\label{prop:GD_convergence}
	Let $\Theta \subset \R^d$ for some $d$. Assume further that for all $x, y, \theta \in \X, \Y, \Theta$, the loss function $\ell(f_\theta(x), y)$ satisfies both $||\nabla_\theta \ell(f_\theta(x), y)||_2^2 \le C$ and
	$$ C_{min}I \preccurlyeq \nabla_\theta^2 \ell(f_\theta(x), y) \preccurlyeq C_{max} I,$$
	where $I$ denotes the $d$-dimensional identity matrix. 
	
	Then, by running gradient descent with step size $\alpha = \frac{1}{C_{max} + 2C \lambda}$, the $k$-th iteration $\theta^{(k)}$ satisfies:
	\begin{equation}
		L(\theta^{(k)}; \lambda) - L(\htl; \lambda) \le \left(1 - \frac{C_{min}}{C_{max} + 2 C \lambda}\right)^k \left(L(\theta^{(0)};\lambda) - L(\htl; \lambda)\right).
	\end{equation}
\end{prop}

As a direct implication, the convergence rate suffers with increases to $\lambda$. However, we might expect that problem \eqref{opt:continuum} remains tractable for up to moderate values of $\lambda$. In addition, note that the required step size to achieve linear convergence is also decreasing in $\lambda$, which may yield further challenges. The design of efficient algorithms to solve problem \eqref{opt:continuum} remains an open area. In our simulations, we observed that computation time was significantly reduced by using $\htl$ as the starting point for finding a new optimum $\hat \theta_{\lambda + \delta}$, for some small step $\delta > 0$.

\section{Experiments}
\label{sec:experiments}
We now demonstrate this paper's methodology by solving problem \eqref{opt:continuum} over a range of $\lambda$ for several common datasets. The following can all be obtained from the UCI Machine Learning Repository \cite{Dua2019UCI}.
\begin{itemize}
	\item \textbf{COMPAS}: Arrest records from 2013 and 2014 in Broward County, Florida by \cite{ProPublicaCOMPAS}, used in \citet{Angwin2016Machine}. 
	\item \textbf{Bank Marketing}: Part of a marketing campaign by a Portugese bank between 2008 and 2013 \cite{Moro2014Data}. 
	\item \textbf{Adult Income}: Collected from the 1994 US Census, including demographic features and income. 
	\item \textbf{Credit Card Default}: Credit card holders of a large Taiwanese bank, collected by \citet{Yeh2009Comparisons}. 
	\item \textbf{Communities \& Crime}: A combination of many different features of counties within the United States, collected by \citet{Redmond2002Data}. Includes sociodemographic data from the Census, survey data from law enforcement, and crime statistics collected by the FBI. 
\end{itemize}
Although the objective function of \eqref{eq:continuum} focuses on maximum individual loss, we also study average losses within groups. Table~\ref{tab:datasets} shows the target variables and the group labels used in each dataset. For conciseness, all other details of our training methodology are omitted, but publicly-available \href{https://github.com/drigobon/rawlsian-algorithmic-fairness}{here}.

\begin{table}
	\caption{Prediction targets and group-defining features.}
	\label{tab:datasets}	
	\centering
	\begin{tabular}{r| l l}
		\toprule
		Dataset & Target & Groups \\
		\midrule
		COMPAS & 2-Year Recidivism & Race \\
		Bank Marketing & Subscription Decision & Marriage Status \\
		Adult Income & Income $>$ \$50,000 & Race \\
		Credit Card Default & Payment Default & Marriage Status \\
		Communities \& Crime & Violent Crime Level & Poverty Percentage (Quartile) \\
		\bottomrule
	\end{tabular}
\end{table}

\begin{figure}
	\hspace*{0.05\linewidth}
	\begin{subfigure}{0.89\linewidth}
		\centering
		\includegraphics[width = \linewidth]{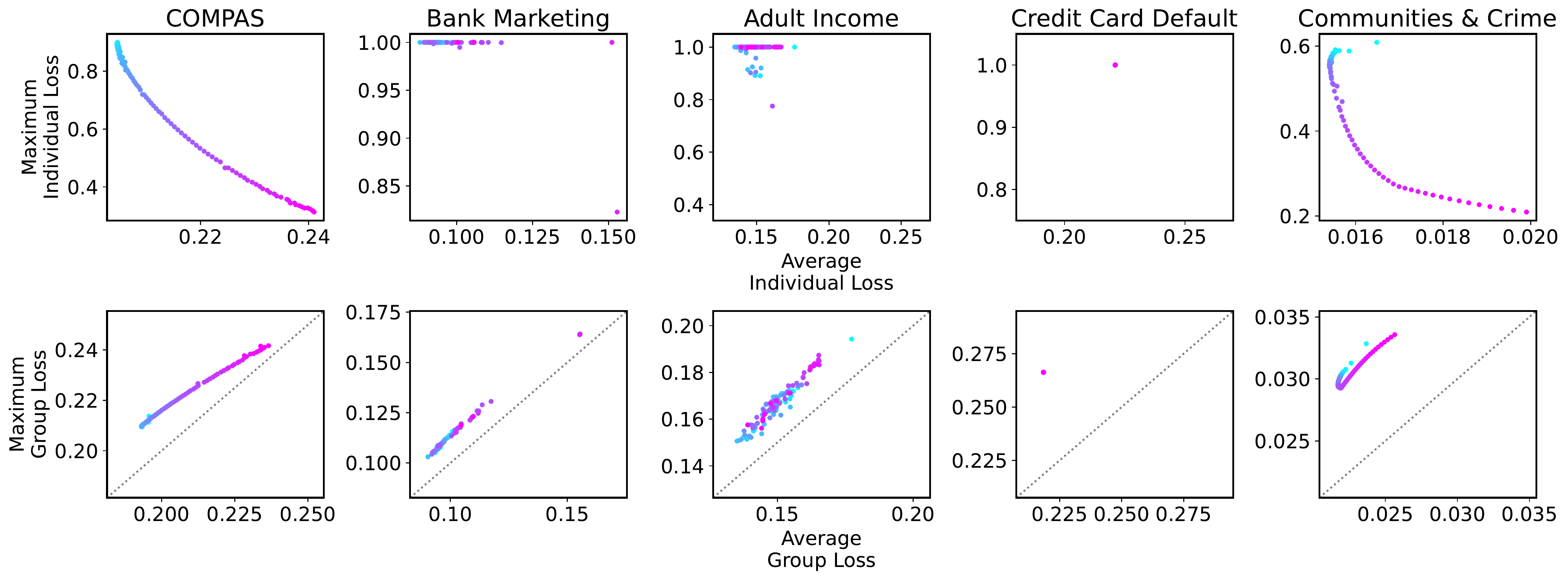}
	\end{subfigure}
	\begin{subfigure}{0.05\linewidth}
		\centering
		\includegraphics[width = 0.75\linewidth]{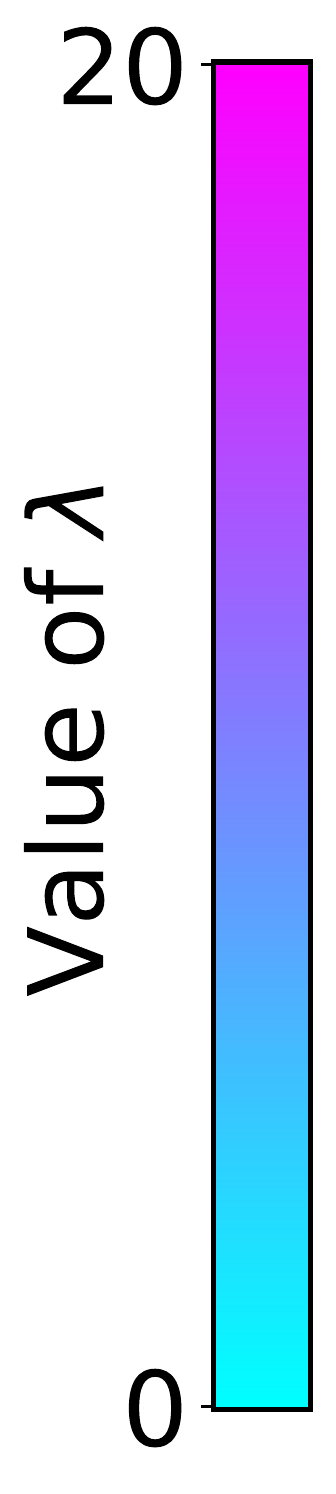}
		\vspace{4.5em}
	\end{subfigure}
	\captionsetup{font=small, width=0.89\linewidth, justification=justified}
	\caption{The tradeoff between average and maximum loss at various $\lambda$ for logistic regression models. Points are colored according to $\lambda$. The top panel illustrates the tradeoff for individual losses. The bottom panel shows the analogue for within-group average loss, with the identity line in gray.}
	\label{fig:logistic}
\end{figure}

\paragraph{Average and Worst-Case Losses}
The tradeoff between average and worst-case loss reflects exactly the tradeoff between utilitarian and Rawlsian measures of the good. Given the sequence of optimal solutions $\{\htl\}_{\lambda > 0}$, we compute their average and worst-case individual loss \textit{within the training sample}. For the simple setting of logistic regression, these are shown in the top panel of Figure~\ref{fig:logistic}. The tradeoff is most visible for the COMPAS and Communities \& Crime datasets, where as $\lambda$ increases we see maximum loss reduced at the expense of average loss. However, for the other datasets worst-case performance is not significantly improved by varying $\lambda$. In fact, it appears that optimal models for the Credit Card Default dataset are indifferent to the value of $\lambda$.

We also compute average- and worst-case group loss for these datasets, where a group's loss is defined as its average -- see \eqref{opt:grp_continuum}. The bottom panel of Figure~\ref{fig:logistic} plots an analogous tradeoff between average and maximum group loss for several values of $\lambda$. In the COMPAS dataset, we see that increasing $\lambda$ yields a gradually more egalitarian outcome -- wherein the average and maximum group losses are increased together. This suggests that equality may come at the expense of all groups.

\paragraph{Increasing Model Complexity}
We are particularly interested in studying how the curves in Figure~\ref{fig:logistic} change as model complexity is increased. Intuitively, this corresponds to enlarging $\Theta$ -- the set of feasible predictive models. Practically, this is associated with a greater degree of model expressibility (e.g. adding additional covariates, or training a model with greater depth). Here, we study neural networks of gradually increasing depth, and compare them to the baseline of a simple logistic regression. The main text only includes results for the COMPAS dataset, with remaining figures found in Appendix~\ref{app:figures}.

The top panel of Figure~\ref{fig:compas} shows that average individual loss is not significantly affected by increasing the number of layers. However, for the same value of average (individual) loss, maximum individual loss can be significantly reduced -- see, for instance, the point with least maximum loss for average loss equal to 0.2. This observation suggests that when increasing model complexity, Rawlsian good may exhibit larger returns than utilitarian good. In the bottom panel of Figure~\ref{fig:compas}, both average and maximum group losses greatly vary. Within this space of group losses, we often see a difference between the egalitarian (i.e. closest to the diagonal) and Rawlsian optimum. This observation suggests that equality remains at odds with both utilitarian and Rawlsian good, and in particular, that a variation of the group-skew condition from \citet{Liang2022Algorithmic} may hold.

\begin{figure}[!ht]
	\begin{subfigure}{\linewidth}
		\centering
		\includegraphics[width = 0.95\linewidth]{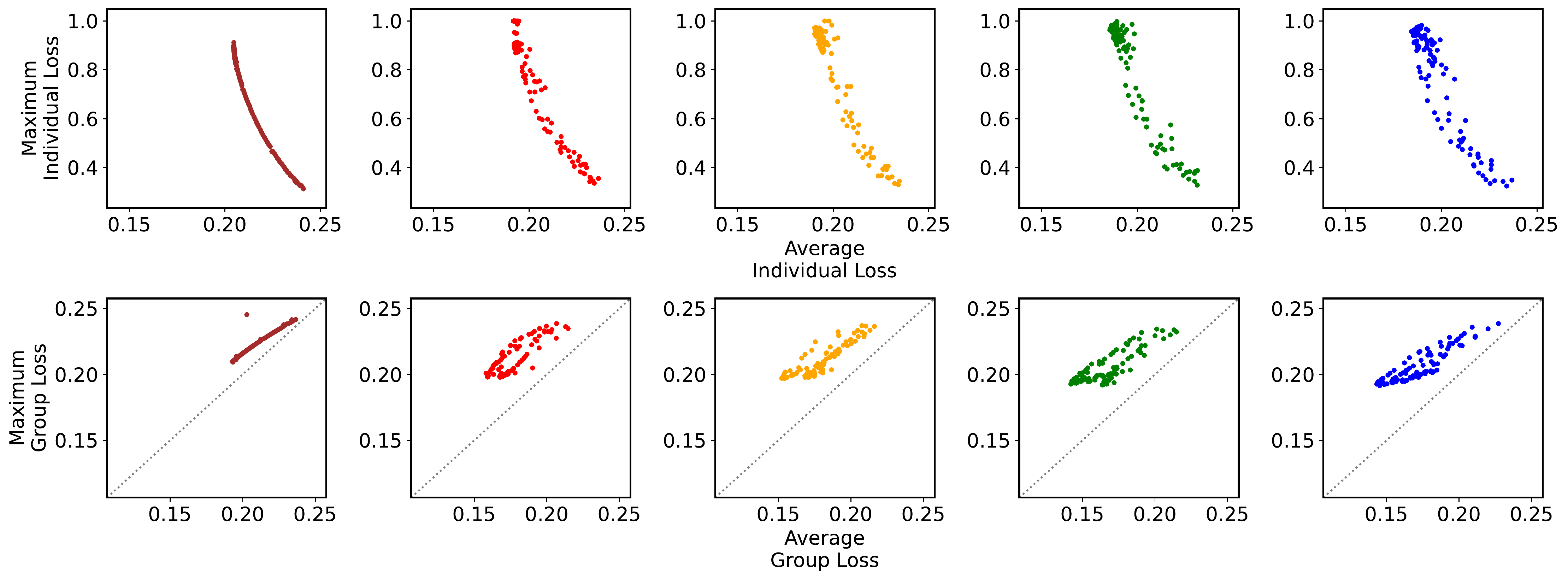}
	\end{subfigure}
	\captionsetup{font=small, width=0.95\linewidth, justification=justified}
	\caption{Result of increasing model complexity for the COMPAS dataset. From left to right: Logistic Regression, 1-Layer Neural Network,..., 4-Layer Neural Network. The top panel plots average vs maximum individual loss. The bottom panel plots average vs maximum group loss, and includes the identity line in gray for reference. In this dataset, groups were defined based on race.}
	\label{fig:compas}
\end{figure}

\section{Discussion and Conclusion}
\label{sec:discussion}

In this paper, we have presented a class of objective functions for supervised learning problems that mixes aspects of both utilitarian and Rawlsian ethical frameworks. Our theoretical results are complimented by experiments on commonly-studied datasets.

Empirically, we often see a tradeoff between utilitarian and Rawlsian measures of good. From an economic perspective, this tradeoff can be interpreted as a `production frontier' between the two goods. In this context, increasing model complexity amounts to greater production capabilities. Therefore, to determine which model along this frontier is best, it is necessary to consider a designer's preferences over fictitious bundles of `utilitarian good' and `Rawlsian good'. Namely, designers must determine how much utilitarian good they are willing to sacrifice for some increase in Rawlsian good. To view minimax fairness (e.g. maximum loss) as a constraint significantly reduces the richness of this question -- effectively assuming that the designer's marginal rate of substitution between these two goods is infinite. Instead, we advocate for a fair-by-design perspective that incorporates broader consideration of a designer's preferences.

The objective functions in this paper correspond to a relaxation of the Rawlsian original position. We have shown that this relaxation is closely tied to expected utility of a risk-averse individual facing random assignment within the population -- a different veil of ignorance. In principle, it is therefore possible to choose an ideal model based on the risk appetite of model-impacted individuals. Hence, we might expect that for low-consequence decisions, risk aversion is low and less emphasis is given to the Rawlsian good. Conversely, high-consequence decision-making such as credit and criminal justice would be strongly influenced by Rawlsian principles. Then clearly, there is unlikely to be a universally agreed upon level of fairness, which must be instead closely tied to a model's use cases.

There are several interesting and valuable directions for future work. First and foremost, we only study two particular approaches to distributive justice and their application to the in-processing stage of model building. It is also common to consider fairness during pre-processing and post-processing stages. Moreover, there are other ethical theories that can inform the development of fair models. For instance, the capability approach in \cite{Sen1999Commodities} was developed as an alternative to utility or resource-based theories of fairness. In addition, we empirically observed a conflict between egalitarian and Rawlsian optima, which has been characterized by \citet{Liang2022Algorithmic} for two groups in classification settings. The further study of these theories and their potential tradeoffs remains an open area of work.

In addition, there are opportunities to further develop theory behind this paper's utilitarian-Rawlsian continuum. For example, it may be possible to develop more efficient algorithms for solving the relaxed optimization problem at large values of $\lambda$, or even for computing the optimal solutions over a wide range of $\lambda$. The statistical properties of these estimators are also of interest, as we believe that large values of $\lambda$ would cause the optimal solution to have large variance (over the randomness of a set of observations). Hence it may be the case that Rawlsian good is at odds with estimation quality. Finally, characterizing the effects of increased model complexity is an extremely interesting open problem. For example, a Rawlsian designer would have no objection to including protected attributes in the training data, as they would be used only to benefit the least advantaged. It is valuable to analyze how much benefit can be gained from doing so.

Human decision-makers are uniquely endowed with the ability to entertain -- but not fully accept -- conflicting ethical perspectives and ideals. We hope that our work is a step towards building models that more closely reflect this ability.

\newpage
\bibliographystyle{plainnat}
\bibliography{library}

\newpage
\appendix

\section{Additional Figures}
\label{app:figures}

For conciseness, Section~\ref{sec:experiments} on increasing model complexity only shows the results for a few datasets. In this appendix, we include the remaining figures, along with other interesting plots.

\begin{figure}[!h]
	\begin{subfigure}{\linewidth}
		\centering
		\includegraphics[width = 0.95\linewidth]{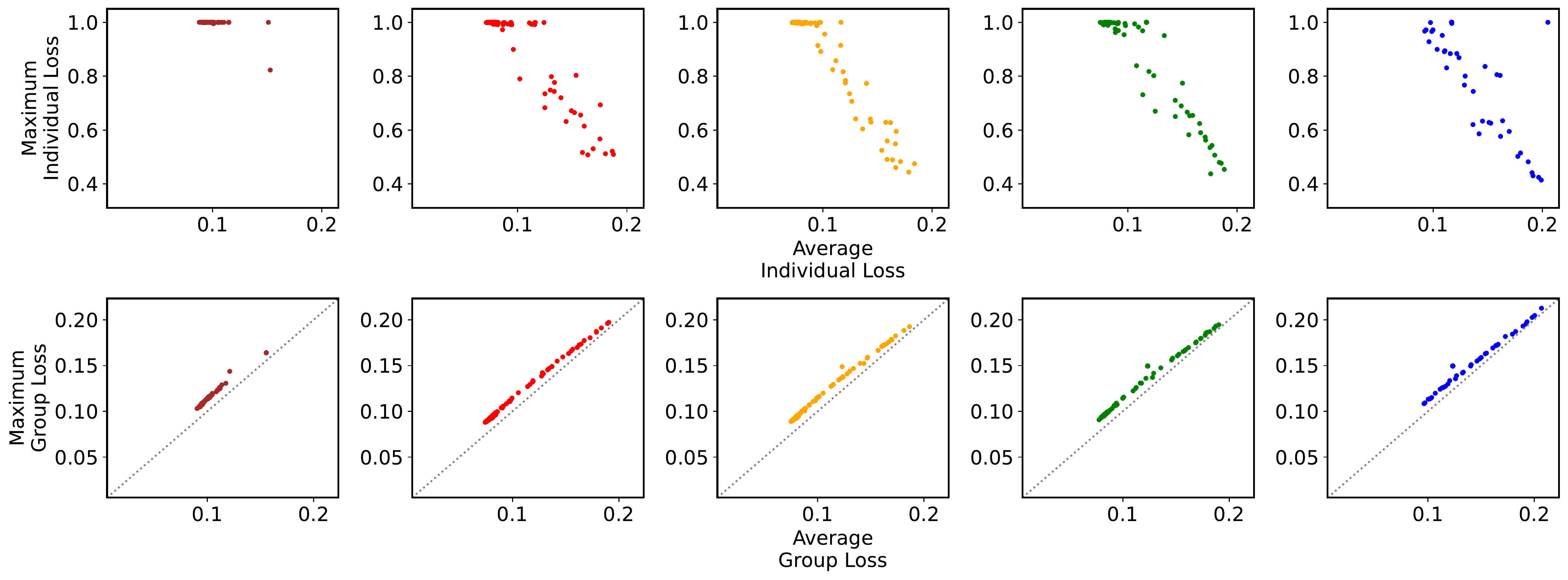}
	\end{subfigure}
	\captionsetup{font=small, width=0.95\linewidth, justification=justified}
	\caption{Result of increasing model complexity for the Bank Marketing dataset. From left to right: Logistic Regression, 1-Layer Neural Network,..., 4-Layer Neural Network. The top panel plots average vs maximum individual loss. The bottom panel plots average vs maximum group loss, and includes for reference the identity line in gray. In this dataset, groups were constructed based on an individual's marital status.}
\end{figure}

\begin{figure}[!h]
	\begin{subfigure}{\linewidth}
		\centering
		\includegraphics[width = 0.95\linewidth]{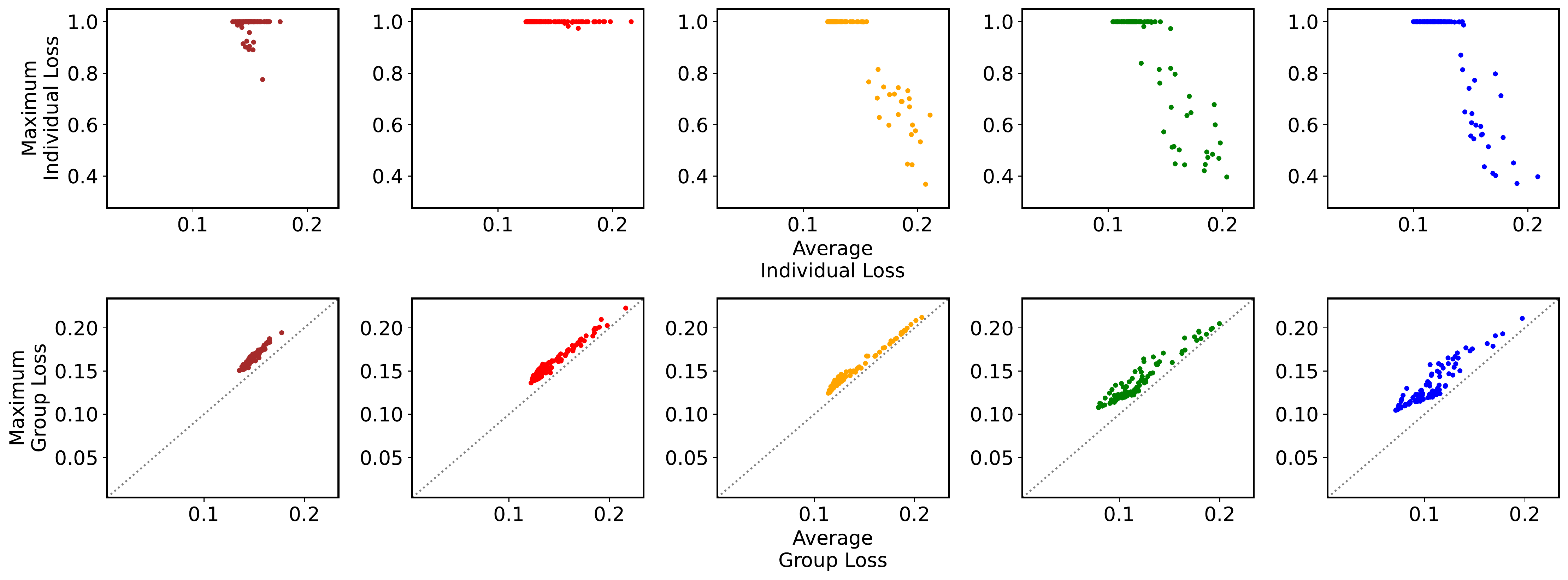}
	\end{subfigure}
	\captionsetup{font=small, width=0.95\linewidth, justification=justified}
	\caption{Result of increasing model complexity for the Adult Income dataset. From left to right: Logistic Regression, 1-Layer Neural Network,..., 4-Layer Neural Network. The top panel plots average vs maximum individual loss. The bottom panel plots average vs maximum group loss, and includes for reference the identity line in gray. In this dataset, groups were constructed based on an individual's race.}
\end{figure}

\begin{figure}[!h]
	\begin{subfigure}{\linewidth}
		\centering
		\includegraphics[width = 0.95\linewidth]{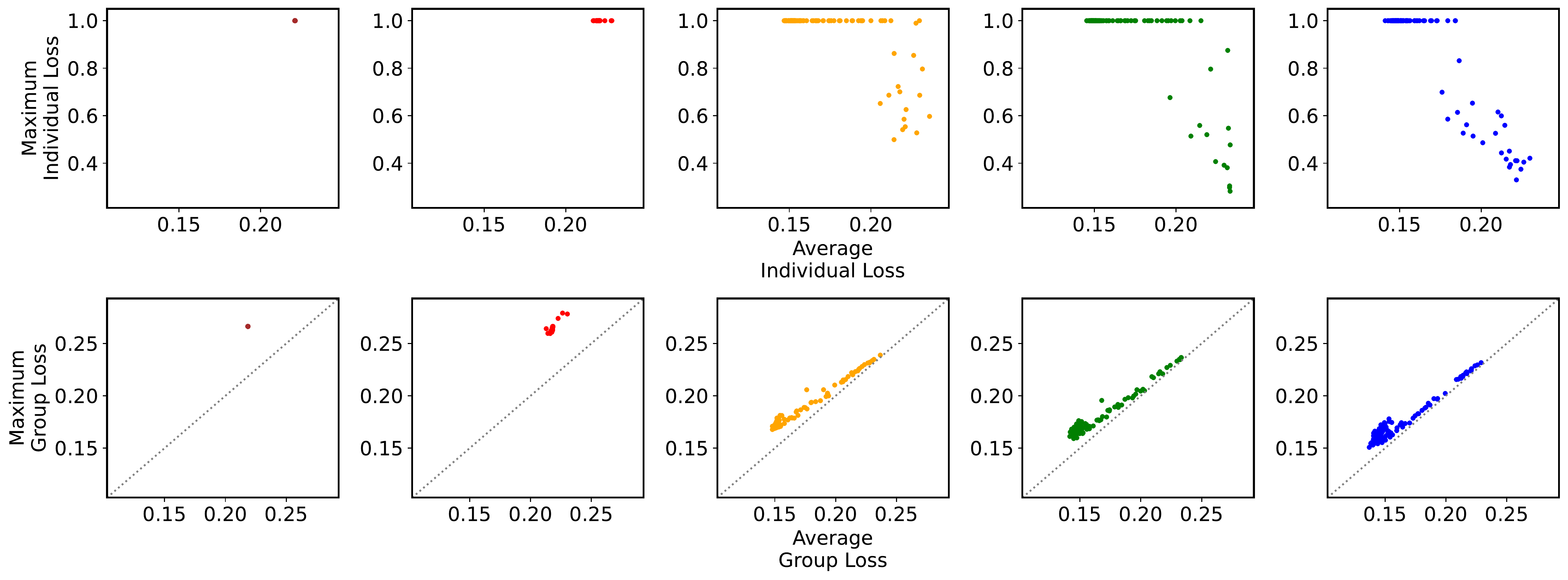}
	\end{subfigure}
	\captionsetup{font=small, width=0.95\linewidth, justification=justified}
	\caption{Result of increasing model complexity for the Credit Card Default dataset. From left to right: Logistic Regression, 1-Layer Neural Network,..., 4-Layer Neural Network. The top panel plots average vs maximum individual loss. The bottom panel plots average vs maximum group loss, and includes for reference the identity line in gray. In this dataset, groups were constructed based on an individual's marital status.}
\end{figure}

\begin{figure}[!h]
	\begin{subfigure}{\linewidth}
		\centering
		\includegraphics[width = 0.95\linewidth]{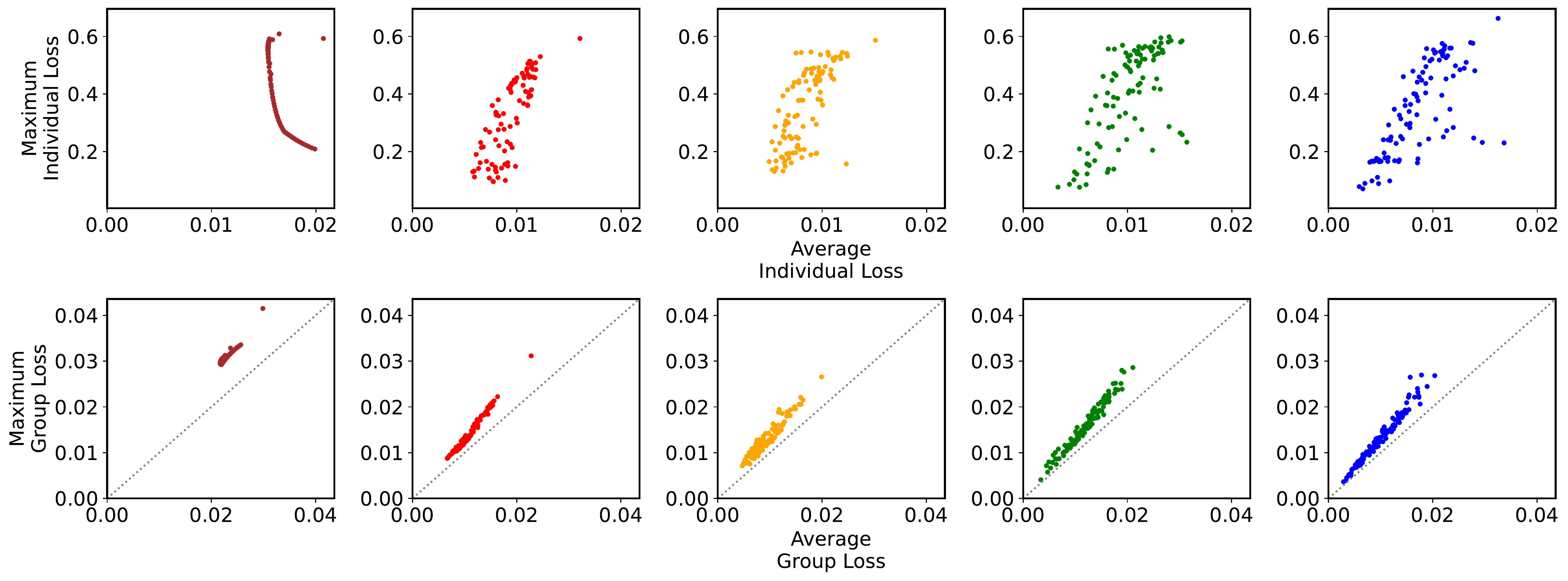}
	\end{subfigure}
	\captionsetup{font=small, width=0.95\linewidth, justification=justified}
	\caption{Result of increasing model complexity for the Communities \& Crime dataset. From left to right: Logistic Regression, 1-Layer Neural Network,..., 4-Layer Neural Network. The top panel plots average vs maximum individual loss. The bottom panel plots average vs maximum group loss, and includes for reference the identity line in gray. In this dataset, groups were constructed based on quartiles of a county's poverty percentage.}
	\label{fig:c&c}
\end{figure}

\pagebreak

\section{Proofs}
\label{app:pfs}

\subsection{Proof of Proposition \ref{prop:pointwise_convergence}}

\begin{proof}
	For convenience of notation, recall that we write $\ell_i(\theta) = \ell(f_\theta(x_i), y_i)$.
	
	First, the limit for $\lambda\to0$ is shown. Taking the limit of the expression directly yields the indeterminate form $\frac{0}{0}$, and applying L'H\^{o}pital's Rule gives:
	$$\lim_{\lambda \to 0} \sum_i \frac{e^{\lambda \ell_i(\theta)} }{\sum_j e^{\lambda \ell_j(\theta)}} \ell_i(\theta),$$
	from which the desired result immediately appears.
	
	\vspace{1em}
	
	When taking the limit $\lambda \to \infty$, another indeterminate form appears, so we again begin with:
	$$\lim_{\lambda \to \infty} \sum_i \frac{e^{\lambda \ell_i(\theta)} }{\sum_j e^{\lambda \ell_j(\theta)}} \ell_i(\theta).$$
	Observe that this is a weighted average, where the $i$-th weight is $\left( \sum_j e^{\lambda (\ell_j(\theta) - \ell_i(\theta)) } \right)^{-1}$. Let us define the set of maximizers $\I = \argmax_i \ell_i(\theta)$, satisfyin $\ell_i(\theta) = \ell_*(\theta)$ for all $i\in \I$. For any $i\notin \I$, notice that $\left( \sum_j e^{\lambda (\ell_j(\theta) - \ell_i(\theta)) } \right) \le e^{-\lambda (\ell_*(\theta) - \ell_i(\theta))}$. The right-hand side converges to zero as $\lambda\to\infty$, and since all weights are lower bounded by zero, this upper bound is tight. The desired limit reduces to
	$$\lim_{\lambda\to\infty} \sum_{i\in \I} \frac{e^{\lambda \ell_i(\theta)}}{\sum_j e^{\lambda \ell_j(\theta)}}   \ell_i(\theta).$$
	By a similar argument, it is possible to see that for all $i\in \I$, $\frac{e^{\lambda \ell_i(\theta)}}{\sum_j e^{\lambda \ell_j(\theta)}} \underset{\lambda\to\infty}{\to} \frac{1}{|\I|}$, and hence:
	$$\lim_{\lambda\to\infty} \sum_{i\in \I} \frac{e^{\lambda \ell_i(\theta)}}{\sum_j e^{\lambda \ell_j(\theta)}}   \ell_i(\theta) = \sum_{i\in \I} \frac{1}{|\I|} \ell_i(\theta)  = \ell^*(\theta)$$
	as desired.	
\end{proof}

\subsection{Proof of Theorem \ref{thm:argmin_convergence}}

\begin{proof}	
	It is useful to define $\L = \{(\ell(f_\theta(x_1), y_1),...,\ell(f_\theta(x_n), y_n)),\, \forall \theta \in \Theta\}$ as the space of all feasible loss profiles. To simplify notation, we write $\ell = (\ell_1,..,,\ell_n)$ to denote an element of $\L$. Since the image of $\Theta$ under $f_\theta(x)$ is assumed to be compact for every $x$, then continuity of $\ell(\cdot, \cdot)$ and compactness of $\Y$ implies compactness of $\L$.

	First, notice that:
	\begin{equation}
		\label{eq:argmin_over_L}
		\argmin_{\ell \in \L} \frac{1}{\lambda} \log \left( \frac{1}{n} \sum_i e^{\lambda \ell_i}\right)  
		=  \argmin_{\ell \in \L} \frac{1}{\lambda n} \sum_i \left(e^{\lambda \ell_i}- 1\right)
	\end{equation}
	for every $\lambda$. In particular, the minimizer of \eqref{eq:argmin_over_L} equals $(\ell(f_{\htl}(x_1),y_1),...,\ell(f_{\htl}(x_n),y_n))$, for $\htl$ solving problem \eqref{opt:continuum}. We define $F_\lambda(\ell) = \frac{1}{\lambda n} \sum_i \left(e^{\lambda \ell_i}- 1\right)$. Differentiating $F_\lambda(\ell)$ with respect to $\lambda$ gives:
	\begin{equation}
		\frac{\partial}{\partial \lambda} F_\lambda(\ell)
		= \frac{\sum_i e^{\lambda \ell_i} \left(e^{-\lambda \ell_i} - (1-\lambda \ell_i) \right)}{\lambda^2 n} \ge 0,
	\end{equation}
	so this sequence of functions is monotone in $\lambda$. In addition, it is easy to show that for any $\ell \in \L$, $\lim_{\lambda \to 0} F_\lambda(\ell) = \frac{1}{n} \sum_i \ell_i$. Since this pointwise convergence holds for a monotone sequence of functions on a compact set, the convergence is uniform in $\L$ \citep[Theorem 7.13]{Rudin1976Principles}:
	\begin{equation}
		F_\lambda(\ell) \overset{\text{unif. in }\L}{\underset{\lambda \to 0}{\to}} \frac{1}{n} \sum_i \ell_i.
	\end{equation}
	Furthermore, the limiting function is continuous in $\ell$, so it follows that this sequence also $\Gamma$-converges in $\L$ (see Theorem 2.1 in \citet{Braides2006Handbook} or Proposition 5.2 in \citet{Maso2012Introduction}). $\Gamma$-convergence can be used to prove that the sequence of minimizers of $F_\lambda$ converges to a minimizer of its $\Gamma$-limit (see Theorem 2.10 in \citet{Braides2006Handbook} or Corollary 7.20 in \citet{Maso2012Introduction}). To apply these results, it is necessary to establish one additional condition on the sequence $\{F_\lambda\}$.  
	
	We say that $\{F_\lambda(\cdot)\}_{\lambda > 0}$ is equi-coercive on $\L$ if for all $t \in \R$ there exists a compact set $K_t$ for which $\{F_\lambda \le t\} \subset K_t$ for all $\lambda$. Since $F_\lambda \ge \frac{1}{n} \sum_i \ell_i$, and the latter has compact sub-level sets on $\L$, then indeed $\{F_\lambda(\cdot)\}_{\lambda > 0}$ is equi-coercive.
	
	Together, equi-coercivity and $\Gamma$-convergence imply that the limit of $\{\hat\ell_\lambda\}_{\lambda > 0}$, the sequence of minimizers to \eqref{eq:argmin_over_L}, is a minimizer to the $\Gamma$-limit of $F_\lambda$. Namely:
	\begin{equation}
		\label{eq:lim_argmin_l_0}
		\lim_{\lambda \to 0} \argmin_{\ell \in \L} F_\lambda(\ell) \in \argmin_{\ell \in \L} \frac{1}{n} \sum_i \ell_i.
	\end{equation}
	To obtain the desired result, it is only necessary to rewrite the optimization problems in terms of $\theta$ and $\Theta$.
	
	Of course, if the minimizer on the right-hand side is unique, then the argmax in \eqref{eq:lim_argmin_l_0} contains only a single value, and it must be that $\lim_{\lambda \to 0} \htl = \htu$.
	
	\vspace{1em}
	
	The proof for taking the limit as $\lambda \to \infty$ is nearly identical. We include its outline here.
	
	Consider now the sequence of functions $G_\lambda(\ell) = \frac{1}{\lambda} \log \left(\sum_i e^{\lambda \ell_i} \right)$. Observe that $G_\lambda$ converges pointwise to $\max_i \ell_i$. Taking a derivative with respect to $\lambda$ gives:
	\begin{equation}
		\frac{\partial}{\partial \lambda} G_\lambda(\ell) 
		= \frac{\lambda \sum_i \frac{e^{\lambda \ell_i} }{\sum_j e^{\lambda \ell_j}} \ell_i   - \log \left( \sum_i e^{\lambda \ell_i} \right)}{\lambda^2} 
		\le \frac{\lambda \max_i \ell_i - \log(\max_i e^{\lambda \ell_i})}{\lambda^2} = 0,
	\end{equation}
	so again this sequence is monotone. Identical arguments imply that $G_\lambda(\ell) \overset{\Gamma}{\underset{\lambda\to \infty}{\to}} \max_i \ell_i$. We can similarly use this sequence's $\Gamma$-limit to construct compact sub-level sets and prove equi-coercivity. So, we obtain
	\begin{equation}
		\lim_{\lambda \to \infty} \argmin_{\ell \in \L} G_\lambda(\ell) \in \argmin_{\ell \in \L} \max_i \ell_i,
	\end{equation}
	and conclude as before.
\end{proof}

\subsection{Proof of Proposition \ref{prop:identifible}}

\begin{proof}
	First, plug in the assumption on $\ell$ and $\epsilon_i = y - f_{\theta^*}(x_i)$. Taking the gradient with respect to $\theta$ we have:
	$$\E\left[\nabla_\theta L(\theta^*; \lambda, X, Y)\right] = \E \left[ \sum_i \frac{e^{\lambda g(\epsilon_i)}}{\sum_j e^{\lambda g(\epsilon_i)}} g'(\epsilon_i) \left(- \nabla f_{\theta^*}(X_i)\right)  \right].$$
	By the tower property we can obtain
	$$\E\left[\nabla_\theta L(\theta^*; \lambda, X, Y)\right] = \E \left[ \sum_i \frac{e^{\lambda g(\epsilon_i)}}{\sum_j e^{\lambda g(\epsilon_i)}}  \left(- \nabla f_{\theta^*}(x_i)\right)  \E\left[g'(\epsilon_i) | g(\epsilon_1),...,g(\epsilon_n), X_i \right] \right].$$
	Recall that $\epsilon_i$ is independent of all $g(\epsilon_j), \ j\neq i$, but $g'(\epsilon_i)$ cannot be pulled out since $g^{-1}$ is not uniquely defined. However, since $g^{-1}(r)$ can only take two values, then $\E\left[g'(\epsilon_i) | g(\epsilon_i) = r,  X_i \right] = 0$ if and only if 
	\begin{equation}
		\label{eq:identifiability_condition}
		g'\left(g^{-1}_{(+)}(r)\right) f_\epsilon\left(g^{-1}_{(+)}(r)\right) + g'\left(g^{-1}_{(-)}(r)\right) f_\epsilon\left(g^{-1}_{(-)}(r)\right) = 0,
	\end{equation}
	for all $r \in \mathrm{Range}(g)$, where we used the notation introduced in the Proposition.
\end{proof}

\subsection{Proof of Proposition \ref{prop:GD_convergence}}

\begin{proof}
	Using Lemma 3 from \cite{Li2021Tilted}, we have:
	\begin{equation}
		\label{eq:hessian_L}
		\begin{split}
			\nabla^2_\theta L(\theta; \lambda) = & \sum_i  \lambda e^{\lambda (\ell(f_\theta(x_i), y_i) - L(\theta; \lambda))}  \big(\nabla_{\theta}\ell(f_\theta(x_i), y_i) - \nabla_\theta L(\theta; \lambda)\big)\big(\nabla_\theta \ell(f_\theta(x_i), y_i) - \nabla_\theta L(\theta; \lambda)\big)^T  \\
			& \qquad+ e^{\lambda (\ell(f_\theta(x_i), y_i) - L(\theta; \lambda))}  \nabla^2_{\theta}\ell(f_\theta(x_i), y_i).
		\end{split}
	\end{equation}
	The largest eigenvalue of this matrix can be upper bounded by Weyl's inequality as follows:
	\begin{align*}
		\lambda_{max}\left(\nabla^2_\theta L(\theta; \lambda)\right) \le & \lambda_{max} \left(\sum_i  \lambda e^{\lambda (\ell(f_\theta(x_i), y_i) - L(\theta; \lambda))}  \big(\nabla_{\theta}\ell(f_\theta(x_i), y_i) - \nabla_\theta L(\theta; \lambda)\big)\big(\nabla_\theta \ell(f_\theta(x_i), y_i) - \nabla_\theta L(\theta; \lambda)\big)^T \right)  \\
		& \qquad+ \lambda_{max} \left(\sum_i  e^{\lambda (\ell(f_\theta(x_i), y_i) - L(\theta; \lambda))}  \nabla^2_{\theta}\ell(f_\theta(x_i), y_i) \right).
	\end{align*}
	The second term can be upper bounded by $C_{max}$, since we assumed that $\nabla_\theta^2 \ell(f_\theta(x),y) \preccurlyeq C_{max} I$ for all $x,y$ and $\theta$. The first can be controlled as follows:
	\begin{align*}
		\lambda_{max} \left(\big(\nabla_{\theta}\ell(f_\theta(x_i), y_i) - \nabla_\theta L(\theta; \lambda)\big)\big(\nabla_\theta \ell(f_\theta(x_i), y_i) - \nabla_\theta L(\theta; \lambda)\big)^T\right) &= || \nabla_{\theta}\ell(f_\theta(x_i), y_i) - \nabla_\theta L(\theta; \lambda) ||_2^2 \\
		&\le || \nabla_{\theta}\ell(f_\theta(x_i), y_i)||_2^2 + ||\nabla_\theta L(\theta; \lambda) ||_2^2 \\
		&\le 2 C,
	\end{align*}
	since we assumed that $||\nabla_\theta \ell||_2^2 \le C$, which itself implies that the norm of $\nabla_\theta L(\theta; \lambda)$ is bounded by the same quantity. Altogether, we arrive at:
	\begin{equation}
		\lambda_{max}\left( \nabla_\theta^2 L(\theta; \lambda) \right) \le C_{max} + 2C \lambda.
	\end{equation}
	
	By dropping the first term in \eqref{eq:hessian_L}, we can also obtain the lower bound of:
	\begin{equation}
		\nabla_\theta^2 L(\theta; \lambda) \succcurlyeq C_{min} I.
	\end{equation}
	Theorem 13 in \cite{Li2021Tilted} concludes.
\end{proof}

\end{document}